\documentclass[11pt,letterpaper]{article}
\usepackage{fullpage}
\usepackage{times}

\usepackage{amsmath}
\usepackage{amsfonts}
\usepackage{amssymb}
\usepackage{amsthm}
\usepackage{amsbsy}
\usepackage{graphicx}
\usepackage{subcaption}
\usepackage{verbatim}
\usepackage{url}
\usepackage{bbm}
\usepackage{multirow}
\usepackage{enumitem}

\usepackage{hyperref}
\usepackage[svgnames]{xcolor}
\usepackage[capitalise,nameinlink]{cleveref}
\hypersetup{colorlinks={true},linkcolor={DarkBlue},citecolor=[named]{DarkGreen}}

\usepackage{bm}
\usepackage{tikz}
\usepackage{ctable}

\usepackage{natbib}

\usepackage{algorithm}
\usepackage{algpseudocode}
\algnotext{EndFor}
\algnotext{EndIf}
\algnotext{EndWhile}
\algnotext{EndProcedure}

\newtheorem{theorem}{Theorem}[section]
\newtheorem{lemma}[theorem]{Lemma}

\newtheorem{proposition}[theorem]{Proposition}
\newtheorem*{claim*}{Claim}

\newcommand{\R}{\mathbb{R}}

\theoremstyle{definition}
\newtheorem{defn}[theorem]{Definition}

\usepackage{libertine}

\usepackage{authblk}

\usepackage{array}
\newcolumntype{?}{!{\vrule width 1pt}}

\newcommand{\bfA}{\mathbf{A}}
\newcommand{\bfB}{\mathbf{B}}
\newcommand{\NW}{\text{NW}}

\newcommand{\xv}{\boldsymbol{\chi}}

\newcommand{\mnw}{\text{MNW}}
\newcommand{\efx}{\text{\normalfont EFX}}
\newcommand{\vefx}{\text{\normalfont vEFX}}

\begin{document}

\allowdisplaybreaks

\title{\bf Maximum Nash Welfare and Other Stories About EFX\thanks{\,This work has been partially supported by the ERC Advanced Grant 788893 AMDROMA ``Algorithmic and Mechanism Design Research in Online Markets'', the MIUR PRIN project ALGADIMAR ``Algorithms, Games, and Digital Markets'', 
the NWO Veni project VI.Veni.192.153,
the ERC Starting Grant 639945 ACCORD ``Algorithms for Complex Collective Decisions on Structured Domains'', and an EPSRC doctoral studentship (Reference 1892947).
We would like to thank Ioannis Caragiannis for fruitful discussions at early stages of this work.
}}

\author[1,2,3]{Georgios Amanatidis}
\author[4]{Georgios Birmpas}
\author[5]{Aris Filos-Ratsikas}
\author[4]{\\ Alexandros Hollender}
\author[1]{Alexandros A. Voudouris}

\affil[1]{University of Essex, U.K.}
\affil[2]{University of Amsterdam, Netherlands}
\affil[3]{Sapienza University of Rome, Italy}
\affil[4]{University of Oxford, U.K.}
\affil[5]{University of Liverpool, U.K.}

\date{}

\maketitle

\begin{abstract}
We consider the classic problem of fairly allocating indivisible goods among agents with additive valuation functions and explore the connection between two prominent fairness notions: maximum Nash welfare (MNW) and envy-freeness up to any good (EFX). We establish that an MNW allocation is always EFX as long as there are at most two possible values for the goods, whereas this implication is no longer true for three or more distinct values. As a notable consequence, this proves the existence of EFX allocations for these restricted valuation functions. While the efficient computation of an MNW allocation for two possible values remains an open problem, we present a novel algorithm for directly constructing EFX allocations in this setting. Finally, we study the question of whether an MNW allocation implies any EFX guarantee for general additive valuation functions under a natural new interpretation of approximate EFX allocations.
\end{abstract}

\section{Introduction}
\label{sec:intro}

\emph{Fair division} refers to the general problem of allocating a set of resources to a set of agents in a way satisfying a desired fairness criterion. A well-known example of such a criterion is \emph{envy-freeness} \citep{GS58,Foley67,Varian74}, 
where each agent perceives the share she receives to be  no worse than what any other agent receives. Since the problem was formally introduced by Banach, Knaster and Steinhaus \citep{Steinhaus48}, fair division has attracted the attention of various scientific disciplines, including mathematics, economics, and political science. During the last two decades, the algorithmic aspects of fair division have been the focus of a particularly active line of work within the computer science community, e.g., see \citep{Procaccia16-survey,BCM16-survey,Markakis17-survey} and references therein. 

We consider the classic setting where the resources are indivisible goods that need to be fully allocated and the agents have additive valuation functions. One of the main challenges in this setting is that classic fairness notions such as equitability, envy-freeness and proportionality---introduced several decades ago having divisible resources in mind---are impossible to satisfy. To see this for envy-freeness, it suffices to consider two agents and one good of value; the agent who does not get the good is going to be envious. This has led to the recent emergence of several weaker fairness notions (see~\nameref{subsec:related}). As a result, there is a plethora of open questions about the existence, the computation and the interrelationships of such notions. In this work we focus on two of the most prominent: envy-freeness up to any good ($ \efx $) and maximum Nash welfare (MNW).

$\efx$, introduced recently by \citet{Gourves2014matroids} and \citet{CaragiannisKMPS19}, is an additive relaxation of envy-freeness. Here an agent may envy another agent but only by the value of the least desirable good in the other agent's bundle. While this added flexibility of $\efx$ takes care of extreme pathological cases like the one mentioned above ($2$ agents, $1$ good), this notion is not well understood yet. Despite the active interest in it, it is not known whether $\efx$ allocations always exist, even for $4$ agents with additive valuation functions.\footnote{The existence of EFX allocations for $2$ agents was shown independently by \citet{Gourves2014matroids} and \citet{PlautR18}, while, very recently,  \citet{chaudhury2020EFX} presented an algorithm that computes an $\efx$ allocation for instances with $3$ agents.} 
We consider the problem of showing the existence of $\efx$ allocations to be one of the most intriguing currently open questions in fair division. 

The \emph{Nash social welfare} (or, simply, Nash welfare) is the geometric mean of the agents' utilities. By considering maximum Nash welfare (MNW) allocations, i.e., allocations that maximize the product of the utilities, we achieve some kind of balance between the \emph{efficiency} of the maximum utilitarian social welfare---the sum of the utilities---and the \emph{individual fairness} of the maximum egalitarian social welfare---the minimum utility. Although not a fairness concept per se, MNW has strong ties to fairness. In the setting where the goods are divisible, each (possibly fractional) MNW allocation corresponds to a \emph{competitive equilibrium from equal incomes}, a market equilibrium (under the assumption that all agents are endowed with the same budget) that is known to guarantee envy-freeness and Pareto optimality \citep{Varian74}. 
Even in our setting, \citet{CaragiannisKMPS19} showed that integral MNW allocations, besides being Pareto optimal, are \emph{envy-free up to one good} (EF1) and approximately satisfy \emph{maximin share fairness} up to a $\Theta(1/\sqrt{n})$ factor, where $n$ is the number of agents.
Both these guarantees are significantly weaker than $\efx$, in the sense that they are both implied by $\efx$ but they do not imply \emph{any} approximation of it. 

In general, MNW does not imply $\efx$. One of our goals is to identify the cases where it does, in terms of the allowed number of distinct values for the goods. For such cases, we immediately obtain that $\efx$ allocations must exist and then investigate how to efficiently compute them, either through maximizing the Nash welfare or directly. Since, in general, MNW does not even imply a non-trivial approximation of $\efx$, we further introduce a less stringent, yet natural, new interpretation of approximate EFX and investigate how it is related to MNW.

\medskip

\subsection{Our contribution} 
There are two variants of $\efx$ used in the related literature, depending on whether only the \emph{positively} valued goods are considered or not; for the latter case we adopt the name $\efx_0$ suggested by \citet{KyropoulouSV19}. We start by establishing a strong algorithmic connection between the two variants (Proposition \ref{prop:efx0-efx}).
Then we explore the relationship between maximizing the Nash welfare and achieving $\efx$ or $\efx_0$ allocations. In doing so, we also obtain some interesting results for the individual notions.
In particular:
\begin{itemize}
\item 
In case there are at most two possible values for the goods (\emph{2-value instances}), we show that any allocation that maximizes the Nash welfare is $\efx_0$ (Theorem \ref{thm:MNW-2v}). This has the following two consequences:
\begin{itemize}
\item 
For any 2-value instance, there exists an $\efx_0$ allocation. Note that this is the first such existence result for non-identical valuations that holds for any number of agents and goods. 
\item 
For the special case of binary valuations, by adapting an algorithm of \citet{BarmanKV18binary}, we can efficiently construct an allocation that is both MNW and $\efx_0$.
\end{itemize}
Note that the implication $\textrm{MNW}\Rightarrow\efx_0$  is no longer true for three or more distinct values.
	
\item While for general 2-value instances the efficient computation of an MNW allocation remains an open problem, we propose a polynomial-time algorithm for producing $\efx_0$ allocations in this case (Theorem \ref{thm:2v-alg}). This algorithm, which we call \textsc{Match\textnormal{\&}Freeze}, is based on repeatedly computing maximum matchings and ``freezing'' certain agents whenever they acquire too much value compared to their peers. We believe these novel ideas might be a stepping stone for proving the existence of $\efx$ allocations in more general settings.

\item We also show that the difficulty of computing $\efx$ allocations does not depend solely on the different number of values, but also on the ratio between the maximum and the minimum value. In particular, for instances where the values of the agents lie in an interval such that the ratio between the maximum and the minimum value is at most $2$, we can compute an $\efx$ allocation using a simple variation of the well-known round-robin algorithm (Theorem \ref{thm:modifiedRR}). 
	
\item For general additive valuations, we show that an MNW allocation does not guarantee any non-trivial approximation of $\efx$. However, we argue that the current definition of \emph{approximate} $\efx$ allocations is not always meaningful. Instead, we explore a different natural definition based on the idea of (hypothetically) augmenting an agent's bundle until an $\efx$-like condition is satisfied. For this new benchmark, which we call \efx-\emph{value}, we show that any MNW allocation is a $1/2$-approximation of $\efx$ (Theorem \ref{MNWapprox}).
\end{itemize}

\medskip

\subsection{Related work} \label{subsec:related}
As there is a vast literature on fair division, here we focus on the indivisible items setting and on related fairness notions. The concept of \emph{envy-freeness up to one good} (EF1) was implicitly suggested by \citet{LMMS04} and formally defined by \citet{Budish11}. \citet{Budish11} also introduced the notion of \emph{maximin share} (MMS), which has been studied extensively  
\citep{KurokawaPW18,AMNS17,BM17,GargMT19,GHSSY18,GargTaki19}
and has yielded several very interesting variants like \emph{pairwise} MMS \cite{CaragiannisKMPS19}, \emph{groupwise} MMS \citep{BBMN18}, and MMS for groups of agents \cite{Suk18}. 

As already mentioned, $\efx$ was introduced by \citet{Gourves2014matroids} (under the term {\em near envy-freeness}) and popularized by \citet{CaragiannisKMPS19}. \citet{PlautR18} defined the notion of $\alpha$-approximate $\efx$ (or $\alpha$-$\efx$) allocations and studied exact and approximate $\efx$ allocations with both additive and general valuations. Most of their results, including the existence of $\efx$ allocations for identical valuations, hold under the similar but stricter notion of $\efx_0$ which is implicitly introduced therein. The currently best $0.618$-approximation of either $\efx$ or $\efx_0$ for the additive case is due to \citet{AMN2020}. For binary additive valuations, \citet{aleksandrov2019greedy} recently proposed an algorithm that produces $\efx$---but not necessarily $\efx_0$---allocations. Independently and at the same time with our work, \citet{Babaioff2020dichotomous} designed an algorithm that computes an $\efx_0$ allocation which maximizes the Nash welfare for submodular dichotomous valuations, a class that includes binary, but does not include general  $2$-value additive valuations.

Very recently, \citet{Manurangsi2020closing} showed that $\efx$ allocations exist with high probability for any number of agents and items under the assumption that the valuations of the agents are drawn at random from a probability distribution. The challenge of showing the existence of $\efx$ allocations is nicely demonstrated in the recent work of \citet{Suksompong2020number}, who showed that in instances with two agents there can be as few as two $\efx$ allocations, while the number of EF1 allocations is always exponential in the number of items.

Besides \citep{CaragiannisKMPS19}, there are several recent papers which relate allocations that maximize (exactly or approximately) the Nash welfare with other fairness notions. \citet{CaragiannisGH19} showed that there exist \emph{incomplete} allocations that are $\efx$ and in which each agent receives at least half of the value they get in a MNW allocation; \citet{CKMS19} achieved the same with only a few unallocated goods. \citet{GargM19} showed how to get an allocation that 2-approximates the Nash welfare of an MNW allocation that is also \emph{proportional up to one good}, satisfies a weak MMS guarantee and is Pareto optimal.

Since computing MNW allocations is an APX-hard problem \citep{Lee2017apx}, there is an active interest on special cases or on approximation algorithms. \citet{BarmanKV18binary} show how to efficiently compute MNW allocations for binary additive valuation functions. 
\citet{ColeG15} were the first to obtain a constant approximation algorithm to the MNW objective. This algorithm, as shown via the improved analysis of \citet{ColeDGJMVY17}, achieves a factor of 2. The currently best-known factor of 1.45 is due to \cite{BarmanMV18Nash}. Going beyond the additive case, in a recent work \cite{GargKK20} study the problem for submodular valuation functions. 


\section{Preliminaries and Notation}
\label{sec:prelim}
We consider fair division instances $I=(N,M,(v_i)_{i \in N})$ in which there is a set $N$ of $n$ {\em agents} and a set $M$ of $m$ indivisible {\em goods}. Each agent $i \in N$ has a {\em valuation function} $v_i: M \rightarrow \mathbb{R}_{\geq 0}$ assigning a non-negative real value $v_i(g)$ to each good $g \in M$. Throughout this work, $v_i$ is {\em additive}, i.e., $v_i(A) = \sum_{g \in A} v_i(g)$ for every set (or {\em bundle}) of goods $A \subseteq M$. 
We pay particular attention to the following subclasses of additive valuation functions: 
\begin{itemize}
\item {\em Binary}: $v_i(g) \in \{0,1\}$ for every $i \in N$ and $g \in M$;
\item {\em $k$-value}: there is a set $V$ consisting of $|V|=k$ distinct, non-negative real values such that $v_i(g) \in V$ for every $i \in N$ and $g \in M$;
\item {\em Interval-value}: for every agent $i \in N$ there exist two real non-negative numbers $x_i$ and $y_i$ such that $x_i < y_i$, and $v_i(g) \in [x_i,y_i]$ for every $g \in M$. 
\end{itemize}
Of course, any binary instance is a $2$-value instance with $V = \{0,1\}$, but we distinguish between these cases as we are able to obtain stronger algorithmic results for the binary case.

A {\em complete allocation} (or just \emph{allocation}) $\bfA=(A_i)_{i\in N}$ is a vector listing the bundle $A_i$ of goods that each agent $i$ receives,  such that $A_i \cap A_j = \varnothing$ for every $i,j \in N$, and $\cup_{i\in N}A_i = M$. Our goal is to come up with allocations that are considered to be {\em fair} by all agents. We begin by defining envy-freeness and its additive relaxations. 

\begin{defn}
\label{def:notions}
An allocation $\bfA = (A_i)_{i \in N}$ is
\begin{itemize}
\item 
{\em envy-free} (EF) if $v_i(A_i)\geq v_i(A_j)$ for every pair $i,j \in N$;

\item 
{\em envy-free up to one good} (EF1) if for every pair $i,j \in N$ with $A_j\neq \varnothing$ there exists a good $g \in A_j$, such that $v_i(A_i) \geq v_i(A_j \setminus \{g\})$;

\item 
{\em envy-free up to any (positively-valued) good} ($\efx$) if for every pair $i,j \in N$ and every good $g \in A_j$ for which $v_i(g)>0$, it holds that $v_i(A_i) \geq v_i(A_j \setminus \{g\})$;

\item 
{\em envy-free up to any good} ($\efx_0$) if for every pair $i,j \in N$ and every good $g \in A_j$, it holds that $v_i(A_i) \geq v_i(A_j \setminus \{g\})$.
\end{itemize}
\end{defn}
\noindent
By definition, we have  $\text{EF} \Rightarrow \efx_0 \Rightarrow \efx \Rightarrow \text{EF1}$, but no implication works in the opposite direction. For brevity, we say that agent $i$ is $\mathcal{E}$ towards agent $j$ when the criterion of $\mathcal{E} \in \{\text{EF},\text{EF1},\efx,\efx_0\}$  is true for the ordered pair $(i,j)$.

As mentioned in the \nameref{sec:intro}, the Nash welfare is usually defined as the geometric mean of the values. Here, for simplicity, we use the \emph{product} of the values instead. As the allocations (exactly) maximizing the Nash welfare are the same under both definitions, this is without loss of generality.

\begin{defn}
\label{def:NWF}
The {\em Nash welfare} of an allocation $\bfA=(A_i)_{i\in N}$ is the product of the values of the agents for their bundles: 
$\NW(\bfA) = \prod_{i\in N} v_i(A_i).$
\end{defn}

We will usually denote by $\bfA^*$ one of the allocations that maximize the Nash welfare (MNW). Among all such allocations, we will sometimes select $\bfA^*$ so that some additional properties are satisfied; e.g., see the discussion in Section~\ref{sec:mnw}. \citet{CaragiannisKMPS19} showed that $\mnw \Rightarrow \text{EF1}$, but the exact connection between MNW and the variants of $\efx$ is not well-understood. 

Before we dive into our main technical results, we show a somewhat surprising connection between $\efx$ and $\efx_0$. In particular, assuming agents with $k$-value valuation functions, for any $k\in \mathbb N$, the question of finding an $\efx_0$ allocation reduces to finding an $\efx$ allocation for an instance with only slightly perturbed valuation functions. 
An immediate corollary is that the existence (resp.~the efficient computation) of $\efx$ allocations for additive agents implies the existence (resp.~the efficient computation) of $\efx_0$ allocations; the converse statements are obvious.

\begin{proposition}\label{prop:efx0-efx}
Let $k\in \mathbb N$. 
The problem of computing $\efx_0$ allocations for $k$-value instances reduces to the problem of computing $\efx$ allocations for $k$-value instances. When all values are rational numbers, this reduction requires only polynomial time.
\end{proposition}

\begin{proof}
Consider any instance $I = (N, M, (v_i)_{i \in N})$. 
Let $\delta$ be the minimum non-zero value difference among any two subsets of goods, according to the valuation function of any agent, that is  
$$\delta = \min_{i \in N} \min_{\substack{A, B \subseteq 2^M: \\ v_i(A) < v_i(B)}}\{ v_i(B) - v_i(A)\}.$$
Furthermore, pick an arbitrarily small $\varepsilon \in (0, \frac{\delta}{m})$.
Now, let $I'=(N,M,(\tilde{v}_i)_{i \in N})$ be an instance such that 
\begin{align*}
\tilde{v}_{i}(g) =
\begin{cases}
v_i(g), & \text{if } v_i(g)>0 \\
\varepsilon, & \text{if } v_i(g)=0.
\end{cases}
\end{align*}
That is, $I'$ is obtained from $I$ by changing any $0$ in the valuation functions of the agents to $\varepsilon$.
Assume that there exists an $\efx$ allocation $\bfA$ for instance $I'$. We will show that $\bfA$ is also an $\efx_0$ allocation for $I$. 

Consider any pair of agents $i, j \in N$. Since $\bfA$ is $\efx$ in $I'$, we have that $\tilde{v}_i(A_i) \geq \tilde{v}_i(A_j \setminus g)$ for every $g \in A_j$. Let $g^* = \arg\min_{g \in A_j} v_i(g)$. Observe that by the choice of $\varepsilon$, $g^*$ is the one that must be ignored when we check whether $\bfA$ is $\efx_0$ for $I$ as well.
\begin{itemize}
\item If $v_i(g^*)=0$, then it must be the case that $v_i(A_i) \geq v_i(A_j)$. Assume otherwise that $v_i(A_i) < v_i(A_j)$. Then, by the definition of $\delta$, it must be $v_i(A_i) \leq v_i(A_j) - \delta$. Using the fact that $v_i(A_j) = v_i(A_j \setminus \{g^*\}) \leq \tilde{v}_i(A_j \setminus \{g^*\})$, by our choice of $\varepsilon$, we have that
\begin{align*}
\tilde{v}_i(A_i) 
\leq v_i(A_i) + m \varepsilon 
\leq v_i(A_j) - \delta + m \varepsilon 
< \tilde{v}_i(A_j \setminus \{g^*\}),
\end{align*} 
contradicting the assumption that $\bfA$ is $\efx$ in $I'$. Hence, $i$ is envy-free towards $j$.

\item
If $v_i(g^*) > 0$, then it must be the case that $v_i(A_i) \geq v_i(A_j \setminus \{g^*\})$. As before, assume otherwise that $v_i(A_i) < v_i(A_j \setminus \{g^*\})$ or, equivalently, $v_i(A_i) \leq v_i(A_j \setminus \{g^*\}) - \delta$. Since 
$ v_i(A_j \setminus \{g^*\}) \leq \tilde{v}_i(A_j \setminus \{g^*\})$, and by our choice of $\varepsilon$, we have that
\begin{align*}
\tilde{v}_i(A_i) 
\leq v_i(A_i) + m \varepsilon 
\leq \tilde{v}_i(A_j \setminus \{g^*\}) - \delta + m \varepsilon  
< \tilde{v}_i(A_j \setminus \{g^*\}),
\end{align*} 
again contradicting the assumption that $\bfA$ is $\efx$ in $I'$. Hence, $i$ is $\efx_0$ towards $j$.
\end{itemize}
Therefore the computation of an $\efx_0$ allocation can be reduced to computing an $\efx$ allocation in an instance with slightly perturbed valuation functions as above. 

In case the values are rational numbers, this reduction needs only polynomial time as $\delta$ is at least $1/D$, where $D$ is the denominator of the product of the values of all agents for all goods, and hence it suffices to choose $\varepsilon \in (0,\frac{1}{mD})$. 
\end{proof}


\section{Maximum Nash Welfare: EFX and Computational Complexity}
\label{sec:mnw}
In this section we focus on allocations that maximize the Nash welfare. We first identify the subclasses of valuation functions for which the MNW allocations are always $\efx_0$, and then consider computational complexity questions. 

Before moving forward, we need to discuss how we handle instances with zero Nash welfare and instances containing {\em zero-valued} goods, i.e., goods for which all agents have value $0$.\footnote{\,Even though it seems quite natural to discard such zero-valued goods, there are settings where one cannot assume {\em free disposal} and all goods must be allocated.}

\medskip

\noindent{\bf Instances with zero Nash welfare.} 
When we talk about the MNW allocations of an instance, the standard interpretation would be to include \emph{all} complete allocations which achieve the maximum Nash welfare. When it is possible to achieve positive Nash welfare this is indeed true. However, for the extreme case of instances where all allocations have zero Nash welfare we are going to need a refinement. Following the work of \citet{CaragiannisKMPS19}, we call an allocation an MNW allocation if it (1) maximizes the number of agents with positive value, and then (2) maximizes the product of the values of such agents. 

The requirements (1) and (2) are by default true for MNW allocations in instances with positive Nash welfare.  
They are also necessary because when the Nash welfare is zero, the idea of maximizing it clearly fails to distinguish ``good'' allocations in any sense.
To illustrate this, consider the next instance:
\begin{center}
\begin{tabular}{c?cccc}
  & $g_1$ & $g_2$ & $g_3$ \\
\noalign{\hrule height 1pt}
agent $ 1 $ & $1$ & $0$ & $0$ \\
agent $ 2 $ & $1$ & $0$ & $0$ \\
agent $ 3 $ & $0$ & $1$ & $1$ \\
\end{tabular}
\end{center}
Since the first two agents only like $g_1$, the Nash welfare of any allocation is $0$. However, not all allocations are $\efx_0$. The allocation $\{\varnothing,\varnothing,\{g_1,g_2,g_3\}\}$ is clearly not $\efx_0$ since the first two agents envy agent $3$ even after the removal of either $g_2$ or $g_3$. Even an allocation such as $\{\{g_1,g_2\},\varnothing,\{g_3\}\}$, which maximizes the number of agents with positive value, is not $\efx_0$ since agent $2$ envies agent $1$ even after the removal of $g_2$. On the other hand, the allocation $\{\{g_1\},\varnothing,\{g_2,g_3\}\}$, which maximizes the number of agents with positive value \emph{as well as} the product of their values, is indeed $\efx_0$: the envy of agent $2$ towards agent $1$ is eliminated by the removal of $g_1$.

\medskip

\noindent{\bf Instances with zero-valued goods.}
While, clearly, zero-valued goods do not affect the Nash welfare of an allocation, they do play an important role as to whether  this allocation is $\efx_0$. To allocate such goods, we first ignore them completely, and compute a Nash welfare maximizing partial allocation $\bfB^*$ only for the remaining goods (which are positively valued by some agent), subject to the requirements (1) and (2) in case the Nash welfare is zero. We then obtain the complete MNW allocation $\bfA^*$ by allocating all the zero-valued goods to one of the agents with the least value according to $\bfB^*$.\footnote{\,For the restricted valuation classes we study here, this suffices. A more general alternative way to complete the allocation would be to allocate all the zero-valued goods to one of the agents that no one envies in $\bfB^*$. It is not hard to show that in any MNW (partial) allocation at least one such agent exists.} 
Observe that $\NW(\bfA^*)=\NW(\bfB^*)$, by definition. Allocating the zero-valued goods this way is also necessary as we illustrate next. Consider the same example as above, but with an extra zero-valued good $g_4$, such that:
\begin{center}
\begin{tabular}{c?ccccc}
  & $g_1$ & $g_2$ & $g_3$ & $g_4$ \\
\noalign{\hrule height 1pt}
agent $ 1 $ & $1$ & $0$ & $0$ & $0$ \\
agent $ 2 $ & $1$ & $0$ & $0$ & $0$ \\
agent $ 3 $ & $0$ & $1$ & $1$ & $0$ \\
\end{tabular}
\end{center}
As before, all allocations have zero Nash welfare, and hence we need an allocation that satisfies (1) and (2). 
The allocation $(\{g_1,g_4\},\varnothing,\{g_2,g_3\})$ is indeed such an allocation: the number of agents with positive value as well as the product of their values are maximized. However, because $g_4$ has been given to agent $1$ (who has value $1$) instead of $2$ (who has value $0$), agent $2$ envies agent $1$ even after its, and thus the allocation is not $\efx_0$. By moving $g_4$ to agent $2$, we obtain the allocation $(\{g_1\},\{g_4\},\{g_2,g_3\})$, which maximizes the number of agents with positive value, the product of their values, and gives the all-zero good $g_4$ to the agent with the least value among all agents, and is indeed $\efx_0$ as agent $1$ has only one good. 

\subsection{When does MNW imply EFX?}\label{subsec:mnw-efx}
Our main result here is that for all $2$-value instances any MNW allocation is also $\efx_0$. Moreover, this result is tight: there exist $3$-value instances for which this implication is no longer true. To simplify the presentation of our results, we distinguish between binary and general $2$-value instances.

\begin{theorem}\label{thm:MNW-binary}
For every binary instance, any MNW allocation is $\efx_0$.
\end{theorem}

\begin{proof}
Consider any binary instance $I = (N,M,(v_i)_{i \in N})$, and let $\bfB^*=(B_i)_{i \in N}$ be the allocation that maximizes the Nash welfare (by maximizing the number of agents with positive value and then the product of their value in case the MNW is zero) for the sub-instance $I_{>0}$ consisting only of the goods which are positively valued by some agent. Then, $\bfA^*$ is obtained from $\bfB^*$ by allocating the remaining goods (which are valued as zero by all agents) to one of the agents with the least value for their own bundles. 

Observe that in $I_{>0}$, $\bfB^*$ must be such that all agents with positive value get goods which they value as $1$ and all agents with zero value get an empty set. Assume otherwise that some agent $i$ gets a good $g$ such that $v_i(g)=0$. Then by moving $g$ to some agent $j \neq i$ with $v_j(g)=1$ we can strictly increase either the product of the values of the agents that have positive value or the number of agents that get positive value, a contradiction. 
Hence, we have that $v_i(B_i) = |B_i|$. 

We next show that by allocating all the zero-valued goods to some agent $i^* \in \arg\min_{i \in N}v_i(B_i)$ we have that $\bfA^*$ is $\efx_0$ as long as $\bfB^*$ is $\efx_0$. 
We distinguish between two cases depending on whether $\NW(\bfB^*) > 0$ or $\NW(\bfB^*)=0$. 

\medskip

\noindent\emph{\underline{Case I:}} $\NW(\bfB^*) > 0$.

\noindent
Consider a pair of agents $i$ and $j$. If $\min_{g \in B_j} v_i(g) = 1$, then $i$ must be $\efx_0$ towards $j$ since $i$ is EF1 towards $j$, and the two notions coincide. 
So, from now on we assume that $\min_{g \in B_j} v_i(g) = 0$. Moreover, we assume that there exists a good in $B_j$ that $i$ values as $1$, since otherwise $i$ would trivially be $\efx_0$ towards $j$, and hence $|B_j| \geq 2$. 

We will show that $i$ is in fact envy-free towards $j$, and thus $v_i(B_i) \geq  v_i(B_j)$.
Assume towards a contradiction that $v_i(B_i) < v_i(B_j)$. Since there exists a good $g \in B_j$ such that $v_i(g)=0$, we have that $v_i(B_j) < v_j(B_j)$. By the fact that $v_i(B_i) = |B_i|$ and $v_j(B_j) = |B_j|$, we thus obtain that $|B_j| \geq |B_i| + 2$. Now, define a new allocation by moving a good in $B_j$ that $i$ values as $1$ from $j$ to $i$. The product of the values of the two agents in the new allocation is equal to  
\begin{align*}
(|B_i| + 1)(|B_j| - 1) = |B_i||B_j| + |B_j| - |B_i| - 1 \geq |B_i||B_j| + 1.
\end{align*} 
Since the bundles of the remaining agents have not changed, the new allocation has strictly higher Nash welfare compared to $\bfB^*$, a contradiction.

\medskip

\noindent\emph{\underline{Case II:}} $\NW(\bfB^*) = 0$. 

\noindent
Consider a pair of agents $i$ and $j$. If $B_i=\varnothing$ and $B_j = \varnothing$, then they are trivially envy-free towards each other. Also, if $B_i\neq \varnothing$ and $B_j \neq \varnothing$, we can show that they are $\efx_0$ towards each other by adapting our arguments for the previous case. Hence, we now focus on the case where $B_i = \varnothing$ and $B_j \neq \varnothing$. If $|B_j|=1$, then $i$ is trivially $\efx_0$ towards $j$. Hence, assume that $|B_j| \geq 2$. We claim that $\max_{g \in B_j} v_i(g) = 0$, and consequently $i$ is envy-free towards $j$. Assume otherwise that there exists a good $g \in B_j$ such that $v_i(g)=1$. Then, by moving $g$ from $j$ to $i$ we can either obtain positive Nash welfare if $i$ is the only agent with zero value, or we can increase the number of agents with positive value in case the Nash welfare remains equal to $0$; since $v_j(B_j)=|B_j| \geq 2$, $j$ still has positive value even after losing $g$. 

\smallskip
	
\noindent In any case, we conclude that $\bfB^*$ is $\efx_0$, and consequently $\bfA^*$ is $\efx_0$ as well. 
\end{proof}

We continue by showing that maximizing the Nash welfare yields an $\efx_0$ allocation for all $2$-value instances. 
Since a $2$-value instance with values $a>1$ and $b=0$ is equivalent to a binary instance (by normalizing the values), Theorem \ref{thm:MNW-binary} above implies that we only need to focus on instances with positive values. Note that in this case $\efx_0$ coincides with $\efx$.

\begin{theorem}\label{thm:MNW-2v}
For any $2$-value instance with positive values, any MNW allocation is $\efx$.
\end{theorem} 

\begin{proof}
Let $a > b > 0$ and consider any $2$-value instance $I = (N,M,(v_i)_{i \in N})$ in which $v_i(g) \in \{a,b\}$ for every $i\in N$ and $g \in M$. Let $i$ and $j$ be any two agents who are given the sets of goods $A_i$ and $A_j$ in an MNW allocation $\bfA^*$. We say that a good is of {\em type} $T_{xy}$ if $i$ and $j$ have values $v_i(g)=x$ and $v_j(g)=y$ for good $g$, respectively; so there are four different types of goods: $T_{aa}$, $T_{ab}$, $T_{ba}$ and $T_{bb}$. If $\min_{g \in A_j} v_i(g) = a$ or $\max_{g \in A_j} v_i(g) = b$, then $i$ is $\efx$ towards $j$ since $i$ is EF1 towards $j$~\citep{CaragiannisKMPS19} and the two notions coincide in this case for the pair $(i,j)$. Therefore, from now on, we will assume that $\min_{g \in A_j} v_i(g) = b$ and $\max_{g \in A_j} v_i(g) = a$, which implies that $|A_j| \geq 2$ and $A_j$ includes at least one good of type $T_{ba}$ or $T_{bb}$.

\medskip
	
\noindent\emph{\underline{Case I:} There is at least one good of type $T_{bb}$ in $A_j$.} 
	
\smallskip
	
\noindent\emph{\underline{Subcase (a):} $A_j$ does not include any good of type $T_{ab}$.}	
Assume, towards a contradiction, that $i$ is not $\efx$ towards $j$: $v_i(A_i)< v_i(A_j) - b $. Since $v_j(g) \geq v_i(g)$ for all  $g \in A_j$, we have that $v_j(A_j) \geq v_i(A_j)$. We now define a new allocation by moving a good $h \in A_j$ of type $T_{bb}$ from $j$ to $i$. In this new allocation, the product of the values of $i$ and $j$ is
\begin{align*}
(v_i(A_i) + b) (v_j(A_j) - b) 
&= v_i(A_i) v_j(A_j) + b(v_j(A_j) - v_i(A_i) - b) \\
&\geq v_i(A_i) v_j(A_j) + b(v_i(A_j) - v_i(A_i) - b) \\
&> v_i(A_i) v_j(A_j).
\end{align*} 
Since the allocation of all other agents has not been changed, the new allocation achieves a strictly larger Nash welfare than $\bfA^*$, yielding a contradiction. 
	
\medskip
	
\noindent\emph{\underline{Subcase (b):} $A_j$ includes at least one good $g$ of type $T_{ab}$.}
We will argue about the structure of set $A_i$. If $A_i$ includes any good $x$ of type $T_{aa}$, $T_{ba}$ or $T_{bb}$, then by exchanging $g$ with $x$, we obtain an allocation with strictly higher Nash welfare, contradicting the choice of $\bfA^*$. For example, if $x$ is of type $T_{aa}$, then in the new allocation (after swapping $x$ and $g$) agent $i$ has exactly the same value, but agent $j$'s value has strictly increased by an amount $a - b > 0$. One can verify that the same holds for the other two types. Hence, $A_i$ must include only goods of type $T_{ab}$, which implies that $v_i(A_i) = |A_i|a$.
	
Towards a contradiction, assume that $i$ is not $\efx$ towards $j$. If $|A_j| \leq |A_i| + 1$, since $A_j$ includes some good $h$ for which $v_i(h)=b$, we have that 
$$v_i(A_j) \leq (|A_j|-1)a + b \leq |A_i|a + b = v_i(A_i) + b,$$ 
i.e., agent $i$ is $\efx$ towards $j$. So, it must be $|A_j| \geq |A_i| + 2$. We create a new allocation by moving a good $g \in T_{ab}$ from $j$ to $i$. The product of the values of $i$ and $j$ then becomes 
\begin{align*}
(v_i(A_i) + a)(v_j(A_j) - b)
= v_i(A_i) v_j(A_j) + a v_j(A_j) - b v_i(A_i) - ab.
\end{align*}
Since $v_j(A_j) \geq |A_j| b \geq (|A_i| + 2) b$ and $v_i(A_i) = |A_i|a$, we have that
\begin{align*}
a v_j(A_j) - b v_i(A_i) - ab 
\geq (|A_i| + 2)ab - |A_i|ab - ab = ab > 0. 
\end{align*}
Since the bundles of the other agents have not been changed, we have that the new allocation has strictly larger Nash welfare than $\bfA^*$, contradicting its choice. 

\medskip
	
\noindent\emph{\underline{Case II:} There are no goods of type $T_{bb}$ in $A_j$.} 

\smallskip
	
\noindent
Then $A_j$ includes at least one good of type $T_{ba}$. 
If $A_j$ includes at least one good of type $T_{ab}$, then, as we argued in Case I(b) above, in order for $\bfA^*$ to be an MNW allocation, $A_i$ cannot include any goods of type $T_{aa}$, $T_{ba}$ or $T_{bb}$. As a result, $A_i$ includes only goods of type $T_{ab}$ and by reproducing the analysis  used in Case I(b) it follows that $\bfA^*$ is $\efx$.
	
So, we may assume that $A_j$ includes goods of type $T_{ba}$ and $T_{aa}$ only. This implies that $v_j(A_j) = |A_j| a$. Assume towards a contradiction that $i$ is not $\efx$ towards $j$: $v_i(A_i) < v_i(A_j) - b$. Since $A_j$ contains at least one good that $i$ values as $b$, we also have that	$v_i(A_j) \leq (|A_j|-1)a + b.$ Combining the last two expressions, we obtain that
$$v_i(A_i) + a < |A_j|a = v_j(A_j).$$
Now, consider the allocation that is obtained from $\bfA^*$ by moving a good of type $T_{aa}$ from $j$ to $i$. We know that such an item exists since $\max_{g \in A_j} v_i(g) =  a$. By using the last inequality, the product of the values of $i$ and $j$ in the new allocation is
\begin{align*}
(v_i(A_i) + a)(v_j(A_j) - a)
&= v_i(A_i) v_j(A_j) + a(v_j(A_j) - v_i(A_i) - a) \\
&> v_i(A_i) v_j(A_j),
\end{align*}
which combined with the fact that the bundles of the other agents have not been changed, contradicts the choice of $\bfA^*$.
	
\smallskip
	
\noindent In any case, we conclude that $\bfA^*$ must be $\efx$.
\end{proof}

\citet{CaragiannisGH19} presented  a $3$-value instance in which no MNW allocation is $\efx$. For completeness, we include here a simpler such instance, which further shows that the implication $\text{MNW} \Rightarrow \{\efx,\efx_0\}$ is no longer true even for interval-value instances in which the length of the interval is almost zero. 
Let $\varepsilon$ be a small positive constant and consider an instance with two agents and three goods with values as shown in the table:
\begin{center}
\begin{tabular}{c?cccc}
 & $g_1$ & $g_2$ & $g_3$ \\
\noalign{\hrule height 1pt}
agent $ 1 $ & $1-\varepsilon$ & $1$ & $1+\varepsilon$ \\
agent $ 2 $ & $1$ & $1-\varepsilon$ & $1+\varepsilon$ \\
\end{tabular}
\end{center}
This is a $3$-value instance with values $\{1-\varepsilon,1,1+\varepsilon\}$. Clearly, it is also an interval-value instance with interval of length $2\varepsilon$, which can be arbitrarily close to zero by selecting $\varepsilon$ to be extremely small. It is easy to verify that there are exactly two allocations achieving the maximum Nash welfare of $2+\varepsilon$: $\bfA_1 = (\{g_2\}, \{g_1,g_3\})$ and $\bfA_2 = (\{g_2,g_3\}, \{g_1\})$. The Nash welfare of any other allocation is either $2(1-\varepsilon)$ or $2+\varepsilon - \varepsilon^2$. Now, for $\ell\in\{1,2\}$, observe that in $\bfA_{\ell}$ agent $\ell$ is not $\efx$ towards the other agent since she envies her even after the removal of $g_{\ell}$.

\subsection{On the complexity of maximizing the Nash welfare}

We now turn our attention to the complexity of computing a maximum Nash welfare allocation. This problem is already known to be hard for many domain restrictions, and easy for only a few special cases. Nevertheless, its complexity for $k$-value instances with $k\in\{2, 3, 4\}$ has been open. Here we make significant progress towards settling these cases. We again start with the binary case. 

\begin{theorem}\label{thm:MNW-binary-alg}
For binary instances, computing an MNW allocation (and thus an $\efx_0$ allocation) can be done in polynomial time.
\end{theorem}

\begin{proof}
Consider any binary instance $I=(N,M,(v_i)_{i \in N})$ and let $I_{> 0}$ be the sub-instance consisting only of the goods that are positively valued by some agent. Given the MNW allocation $\bfB^*$ for $I_{>0}$, we can obtain the MNW allocation $\bfA^*$ for $I$ by augmenting $\bfB^*$ so that all the zero-valued goods are given to the agent with minimum value according to $\bfB^*$. So, it remains to compute $\bfB^*$ in $I_{>0}$.  

To do this, we use the greedy algorithm \textsc{Alg-Binary} of \citet{BarmanKV18binary} which outputs an allocation maximizing the Nash welfare for the binary instance. Let $\bfB$ be the allocation that \textsc{Alg-Binary} outputs when given as input $I_{> 0}$. If $\NW(\bfB) > 0$, then $\bfB^*=\bfB$ is also $\efx_0$ by Theorem~\ref{thm:MNW-binary}. 
However, if $\NW(\bfB)=0$, in which case all allocations have zero Nash welfare, $\bfB$ might {\em not} be an MNW allocation in our sense, i.e., an allocation that maximizes the number of agents with positive value and then the product of their values. Hence, $\bfB$ may not be $\efx_0$. 
To circumvent this, we define a bipartite graph consisting of nodes corresponding to the agents on the left and nodes corresponding to the goods on the right, while an edge between an agent and a good exists if the agent has value $1$ for the good. By computing a maximum bipartite matching on this graph, it is guaranteed that the number of agents with positive value is maximized. Then, we run {\sc Alg-Binary} on the restricted sub-instance of $I_{>0}$ where the set of agents includes only the ones that participate in the maximum matching, so that the product of their values (which now is going to be positive) is also maximized. This yields the desired allocation $\bfB^*$ with maximum Nash welfare for $I_{>0}$ in which the agents that did not participate in the maximum matching get an empty set. 
\end{proof}

For general $2$-value instances we were unable to resolve the complexity of computing an MNW allocation, but we show that the problem is NP-hard for $3$-value instances. This extends the hardness aspect (but not the inapproximability) of the result of \citet{Lee2017apx} for $5$-value instances. 

\begin{theorem}\label{thm:MNW-hard}
Computing an MNW allocation is \emph{NP}-hard, even for $3$-value instances.
\end{theorem}

\begin{proof}
We will prove that the problem of deciding whether there exists an allocation that achieves Nash welfare at least some value $U$ is NP-complete. Given an allocation, it is trivial to check whether its Nash welfare is at least $U$. For the hardness, we give a reduction from a special version of 3SAT, known as 2P2N-3SAT, where every variable appears twice as a positive literal and twice as a negative literal. 
This problem is known to be NP-complete~\citep{Yoshinaka2005,Berman20032p2n3sat}.
	
Consider an instance of 2P2N-3SAT, in which the set of variables is $\{x_1, \dots, x_n$\}, and the set of clauses is $\{C_1, \dots, C_m\}$. We will now describe how to construct a $3$-value instance with $2m+5n$ goods and $3m+2n$ agents, where the set of values $V$ consists of $a$, $b=1$ and $c=0$, where $a> 1/(\sqrt[2m]{2}-1)$.
	
For each variable $x_i$, introduce two variable-agents $\{T_i,F_i\}$, as well as $5$ variable-goods, denoted as $\{s_{i,0},s_{i,1},s_{i,2},s_{i,3},s_{i,4}\}$. The values of the variable-agents for the variable-goods are:
\begin{itemize}
\item $s_{i,0}$: both $T_i$ and $F_i$ have value $a$ for it; the value of all other variable-agents is $0$.
\item $s_{i,1}$ and $s_{i,2}$: $T_i$ has value $1$ for each of them; the value of all other variable-agents is $0$. 
\item $s_{i,3}$ and $s_{i,4}$: $F_i$ has value $1$ for each of them; the value of all other variable-agents is $0$. 
\end{itemize}
	
For each clause $C_j = (\ell_1 \lor \ell_2 \lor \ell_3$), where $\ell_1$, $\ell_2$ and $\ell_3$ are the three literals in the clause, introduce three clause-agents $\{C_j^1, C_j^2, C_j^3\}$ and two clause-goods $\{p_j,q_j\}$.
The values of the clause-agents for all goods as well as the values of the variable-agents for the clause-goods are as follows:
\begin{itemize}
\item Both $p_j$ and $q_j$ are valued as $a$ by the three corresponding clause-agents $C_j^1, C_j^2$ and $C_j^3$, and as $0$ by all other (variable- and clause-) agents.
		
\item For every $j$, the clause-agent $C_j^t$, $t \in [3]$ has value $0$ for all other goods, besides one:
\begin{itemize}
\item If $\ell_t = x_i$, then $C_j^t$ has value $1$ for one of two variable-goods $s_{i,1},s_{i,2}$ (whichever is not valued by some other clause-agent).
			
\item If $\ell_t = \overline{x_i}$, then $C_j^t$ has value $1$ for one of the two variable-goods $s_{i,3},s_{i,4}$ (whichever is not valued by some other clause-agent).
\end{itemize}
\end{itemize}
Note that since there are exactly two occurrences of $x_i$ and $\overline{x_i}$ in the 2P2N-3SAT instance, each of the variable-goods $s_{i,1},s_{i,2},s_{i,3},s_{i,4}$ is valued by exactly two agents (and one of them is always $T_i$ or $F_i$).
	
Given a satisfying assignment for the 2P2N-3SAT instance, we define the following allocation with Nash welfare at least $U=2^n a^{2m+n}$:
\begin{itemize}
\item Variables: If $x_i=1$, then we allocate $s_{i,0}$ to $T_i$ (for value $a$) and $\{s_{i,3},s_{i,4}\}$ to $F_i$ (for value $2$). Otherwise ($x_i=0$), we allocate $s_{i,0}$ to $F_i$ (for value $a$) and $\{s_{i,1},s_{i,2}\}$ to $T_i$ (for value $2$). 
		
\item Clauses: Let $C_j =(\ell_1 \lor \ell_2 \lor \ell_3)$. Since $C_j$ is satisfied, at least one of $\ell_1$, $\ell_2$ or $\ell_3$ is true; without loss of generality, assume that $\ell_1$ is true. Then, we allocate $p_j$ to $C_j^2$ (for value $a$), and $q_j$ to $C_j^3$ (again for value $a$). Note that there is now exactly one good valued by $C_j^1$ that is still unallocated. If $\ell_1=x_i$, then this good is one of $s_{i,1}$ or $s_{i,2}$, while if $\ell_1=\overline{x_i}$, then this good is one of $s_{i,3}$ or $s_{i,4}$. In any case, we allocate this good to $C_j^1$.
		
\item We allocate any remaining goods arbitrarily.
		
\end{itemize} 
Now observe that for any variable $x_i$ the product of the values of the two variable-agents $T_i$ and $F_i$ is at least $2a$, and for any clause $C_j =(\ell_1 \lor \ell_2 \lor \ell_3)$ the product of the values of the three clause-agents $C_j^1$, $C_j^2$ and $C_j^3$ is at least $a^2$. Consequently, the Nash welfare of the resulting allocation is at least $(2a)^n (a^2)^m = 2^na^{2m+n}$.
	
\medskip 
	
Conversely, we will now show that given any allocation with Nash welfare at least $2^na^{2m+n}$, we can obtain a satisfying assignment for the 2P2N-3SAT instance. Without loss of generality, we can assume that every good has been allocated to some agent that has positive value for it. Otherwise, we can modify the allocation so that this holds and the Nash welfare will not decrease. In particular, in any clause $C_j$, the goods $p_j$ and $q_j$ will be allocated to the agents $C_j^1, C_j^2, C_j^3$.
	
More specifically, we can assume that $p_j$ and $q_j$ have been allocated to distinct clause-agents. Suppose otherwise that $p_j$ and $q_j$ have both been allocated to the same agent, say $C_j^1$. If one of $C_j^2$ or $C_j^3$ has value $0$ (and hence the Nash welfare is $0$), then by allocating $q_j$ to that agent instead the Nash welfare will not decrease. If both $C_j^2$ and $C_j^3$ have positive value, then this must be equal to $1$. Reallocating $q_j$ to (say) $C_j^2$ will not decrease the Nash welfare. Indeed, before the modification the product of values was at most $2a+1$ and after it is at least $a(a+1)$. By our choice of $a$, we have $2a+1 \leq a(a+1)$, so the Nash welfare cannot decrease as a result of this modification.
	
By the observation above, we can assume that $s_{i,0}$ has been allocated to $T_i$ or $F_i$. If it has been allocated to $T_i$, then we can also assume that goods $s_{i,1}$ and $s_{i,2}$ have not been allocated to $T_i$, but instead that they have each been allocated to the unique other agent that values them. Let us show this for $s_{i,1}$ (the other case is identical). Assume that $s_{i,1}$ has been allocated to $T_i$ and let $C_j^t$ be the unique other agent that values it. By the above discussion, $C_j^t$ has obtained at most one of $p_j$ and $q_j$. If $C_j^t$ has value $0$, it is easy to see that allocating $s_{i,1}$ to $C_j^t$ does not decrease the Nash welfare. If $C_j^t$ has positive value, then this must be equal to $a$. Now if $T_i$ has utility $a+2$, then by allocating $s_{i,1}$ to $C_j^t$ instead, the product of values increases from $(a+2)a$ to $(a+1)^2$. On the other hand, if $T_i$ has utility $a+1$, then by allocating $s_{i,1}$ to $C_j^t$ instead, the product of values remains the same (it goes from $(a+1)a$ to $a(a+1)$). Similarly, if good $s_{i,0}$ has been allocated to $F_i$, we can assume that goods $s_{i,3}$ and $s_{i,4}$ have not been allocated to $T_i$, but instead that they have each been allocated to the unique other agent that values them.
	
We now construct an assignment for the variables as follows:
\begin{itemize}
\item If $T_i$ has obtained $s_{i,0}$, then we set $x_i=1$. 
\item If $F_i$ has obtained $s_{i,0}$, then we set $x_i=0$.
\end{itemize}
We say that the allocation satisfies the \emph{consistency} property, if the agent in $\{T_i,F_i\}$ that has not obtained $s_{i,0}$, has obtained all other goods which she values for a total value of $2$. 
	
Let us first show that if the allocation satisfies the consistency property, the assignment satisfies the 2P2N-3SAT instance. Observe that for every clause $C_j$, one of the three agents $C_j^1$, $C_j^2$ or $C_j^3$ will not obtain one of the two clause-goods which they all value as $a$. Thus, this agent has to obtain the single other variable-good which she positively values, otherwise the Nash welfare would be $0$. By the consistency property, it follows that the literal associated to this agent must be satisfied. Thus, the clause is satisfied.
	
Finally, let us show that if the Nash welfare of the allocation is at least $2^na^{2m+n}$, then the consistency property holds. Assume that there exists an $i$ such that the agent in $\{T_i,F_i\}$ that has not obtained $s_{i,0}$, has value strictly less than $2$, i.e.\ at most $1$. Thus, the product of the values of $T_i$ and $F_i$ is at most $a$. For any other $i'$, the product of the values of $T_{i'}$ and $F_{i'}$ is at most $2a$. For any clause $C_j$, the product of the values of the three corresponding clause-agents is at most $(a+1)^2$. Thus, altogether the Nash welfare is at most $(a+1)^{2m} (2a)^{n-1}a = 2^{n-1}(a+1)^{2m}a^n$. By our choice of $a > 1/(\sqrt[2m]{2}-1)$, we have ensured that $(a+1)^{2m} < 2a^{2m}$, and so the Nash welfare is indeed strictly less than $2^na^{2m+n}$.
\end{proof}


\section{Computing EFX Allocations for Restricted Domains}\label{sec:restricted}
Even though we showed that any MNW allocation is also $\efx_0$ for agents with $2$-value valuation functions, it remains an open question whether there exists a polynomial-time algorithm for computing such allocations beyond the binary case. In this section, we try to circumvent this and aim to design efficient algorithms for computing $\efx_0$ allocations (which might not maximize the Nash welfare) for $2$-value instances and interval-value instances.

\subsection{2-value instances}
We begin with considering $2$-value instances with values $\{a,b\}$ such that $a > b \geq 0$. 
Our algorithm, which we call {\sc Match\textnormal{\&}Freeze}, proceeds in rounds and maintains a set of \emph{active} agents $L$, initially containing everyone. In each round, every active agent is given exactly one of the remaining goods, with the possible exception of the last round in which there might not be enough goods left for all agents.
The algorithm terminates when all goods have been allocated.  

\begin{algorithm}[h!]
	\caption{{\sc Match\textnormal{\&}Freeze}$(N,M,(v_i)_{i \in N})$}
	\begin{algorithmic}[1] 
		\State Input: a 2-value instance using the values $a,b$ ($a > b \geq 0$)
		\State $L \leftarrow N$ \Comment{set of active agents}
		\State $R \leftarrow M$ \Comment{set of unallocated goods}
		\State $\ell=(1,2,\ldots,n)$ \Comment{ordered list of agents}
		\While{$R \neq \emptyset$} \Comment{every iteration is a round}
		\State Construct the bipartite graph $G = (L \cup R,E)$.
		\State Compute a maximum matching on $G$. 
		\For{each matched pair $(i,g)$}
		\State Allocate good $g$ to agent $i$.
		\State Remove $g$ from $R$.
		\EndFor
		\For{each unmatched active agent $i$ w.r.t.~$\ell$}
		\State Allocate \emph{one} arbitrary unallocated good $g$ to $i$.
		\State Remove $g$ from $R$.
		\EndFor
		\State Construct the set $F$ of agents that need to freeze.
		\State Remove agents of $F$ from $L$ for the next $\lfloor a/b - 1 \rfloor$ rounds.
		\State Put agents of $F$ to the end of $\ell$.
		\EndWhile
		\State \Return the resulting allocation $\bfA$.
	\end{algorithmic}
	\label{alg:efx-2-value}
\end{algorithm}

To determine which good each active agent gets during a round, we create a bipartite graph $G=(L \cup R, E)$ with nodes corresponding to the active agents $L$ on one side and to the remaining goods $R$ on the other. 
An edge between an active agent $i$ and a good $g$ exists if and only if $v_i(g) = a$.
We first compute a maximum matching on this graph. Then each agent gets the good to which she is matched. If there are agents who are not matched to any good and there are still available goods, the unmatched agents receive one arbitrary available good each (subject to availability). 

There are two possible reasons why an agent $i$ is not matched to any good in a round: (1) she does not have value $a$ for any good (only $b$), or (2) the maximum matching is such that all goods for which her value is $a$ are given to other agents. 
Case (1) does not affect whether the final allocation will be $\efx_0$, but case (2) is  crucial. This is because agent $i$ might now have much smaller value for her own bundle compared to her value for the bundles of some agents that \emph{just received} one good each that $i$ values as $a$.
Let $Z$ be the set of these agents. To make up the distance, agent $i$ should possibly receive multiple goods of value $b$ while all agents in $Z$ must {\em freeze} 
for a number of subsequent rounds depending on the ratio $a/b$.

We define the set $F$ of agents that need to freeze at the end of round $r$ to consist of all those agents who must become inactive  because they have obtained too much value from the perspective of other agents (similarly to Case (2) above). 
Formally, for every active agent $i$, let $g_i$ be the good she gets in round $r$. We begin by setting $F =  \{i \in L \,|\, \allowbreak\exists j \in L : v_j(g_i) = a, v_j(g_j) = b\}$. Then, iteratively, as long as there is an agent $i \in L \setminus F$ such that there exists $j \in F$ with $v_j(g_i) = a$, we also add $i$ to $F$. Each agent in $F$ will remain \emph{frozen} for the next $\lfloor a/b - 1 \rfloor$ rounds. In the case $b=0$, we use the interpretation $\lfloor a/b - 1 \rfloor = + \infty$ in which case agents in $F$ remain frozen forever. Exploiting the properties of the maximum matchings used to allocate goods we can prove that no agent in $F$ will become envious while frozen, and that $F$ is always a strict subset of $L$. The latter means that there is at least one (non-frozen) active agent at any time, and thus the algorithm will terminate after at most $m$ rounds. 

\begin{theorem}\label{thm:2v-alg}
For any $2$-value instance, {\sc Match\textnormal{\&}Freeze} computes an $\efx_0$ allocation in polynomial time. 
\end{theorem}

\begin{proof}
Let $\bfA = (A_i)_{i \in N}$ be the allocation outputted by {\sc Match\textnormal{\&}Freeze} when given as input a $2$-value instance $I = (N,M,(v_i)_{i \in N})$ such that $v_i(g) \in \{a,b\}$, $a > b \geq 0$, for every $i \in N$ and $g \in M$. 
For any agent $i$, let $r_i$ be the round in which the last goods of value $a$ for agent $i$ were allocated. We will need the following two lemmas.

\begin{lemma}\label{clm:from-matching}
For every agent $i$, it holds that:
\begin{itemize}
\item She was allocated a good for which she has value $a$ in each of the rounds $1,2, \dots, r_i-1$.
\item She can freeze only at the end of round $r_i$, and only if during that round she got a good for which she has value $a$. \item She can freeze at most once. After freezing, she has value $b$ for each of the remaining goods.
\end{itemize}
\end{lemma}

\begin{proof}[Proof of Lemma \ref{clm:from-matching}]
For the first part of the lemma, assume that agent $i$ was allocated a good she values as $b$ in some round $r < r_i$. This would mean that agent $i$ was left unmatched in round $r$. However, in round $r_i > r$ there were still goods that $i$ values as $a$ available. Consequently, agent $i$ could also have been matched in round $r$, which means that we did not use a maximum matching, a contradiction.

For the second and third part of the lemma, consider an agent $i$ that got frozen at the end of some round $r$. Since agent $i = i_0$ got frozen, there exist distinct agents $i_1, \dots, i_k$ ($k \leq n$) such that $v_{i_{\ell}}(g_{i_{\ell-1}}) = a$ for $\ell \in [k]$, and $v_{i_k}(g_{i_k}) = b$. Note that agent $i_k$ was not matched in round $r$. Let us now show that the following cases are impossible:

\begin{itemize}
\item $v_i(g_i) = b$. 
This means that agent $i=i_0$ was also not matched in round $r$, and out of the agents $i_0, \dots, i_k$ at most $k-1$ were matched in round $r$. However, since $v_{i_{\ell}}(g_{i_{\ell-1}}) = a$ for every $\ell \in [k]$, we could match agent $i_\ell$ to good $g_{\ell-1}$, and obtain a matching of size $k$ rather than at most $k-1$, a contradiction. Therefore, it must be $v_i(g_i) = a$.

\item $r > r_i$. 
By the definition of $r_i$, agent $i$ must have been allocated a good she values as $b$ in round $r$, which is impossible as we showed above. Hence, it must be $r \leq r_i$.

\item $r < r_i$.
By the definition of $r_i$, at the end of round $r$ there exists an unallocated good $g^*$ such that $v_i(g^*)=a$. Since $v_{i_k}(g_{i_k})=b$, $k$ of the agents $i_0, \dots, i_k$ were matched in round $r$. However, by matching agent $i_\ell$ to good $g_{\ell-1}$ for each $\ell \in [k]$, and $i=i_0$ to $g^*$, we can obtain a larger matching of size $k+1$ rather than $k$, a contradiction. Therefore, it must be $r \geq r_i$. 
\end{itemize}
By the last two cases, we have that agent $i$ can only get frozen in round $r=r_i$, and by the definition of $r_i$, the value of $i$ for any of the remaining goods is $b$. 
\end{proof}

\begin{lemma}\label{clm:ef-to-efx}
If at the beginning of some round $r > r_i$, an active agent $i$ is envy-free towards agent $j$, then agent $i$ will be $\efx_0$ towards agent $j$ at the end of the algorithm.
\end{lemma}

\begin{proof}[Proof of Lemma \ref{clm:ef-to-efx}]
By the second part of Lemma~\ref{clm:from-matching}, agent $i$ is active during all rounds after $r > r_i$ until the end of the algorithm, and $i$ has value $b$ for all remaining goods. Consequently, agent $i$ will be allocated a good she values as $b$ in each subsequent round, except potentially the last round (during which she may not get any good), while the value of $i$ for the bundle of agent $j$ can increase by at most $b$ in each subsequent round as well. Consequently, when the algorithm terminates, $v_i(A_j)$ can be at most $b$ more than $v_i(A_i)$. 
\end{proof}

It is easy to see that the algorithm terminates in polynomial time. If no agent ever gets frozen, then the algorithm terminates after at most $\lceil m/n \rceil$ rounds. Otherwise, let $r$ be the first round in which an agent gets frozen, i.e., there exists some agent $j$ with $v_j(g_j) = b$ and some agent $i$ with $v_j(g_i)=a$. By the definition of $r_j$ and Lemma~\ref{clm:from-matching}, it follows that $r=r_j$. Furthermore, again by Lemma~\ref{clm:from-matching}, we know that agent $j$ did not get frozen in round $r=r_j$ (because $v_j(g_j) = b$), which means that she never gets frozen. Since $j$ gets a good in every round, the algorithm terminates after at most $m$ rounds.

Let us now show that the algorithm constructs an $\efx_0$ allocation. If some agent $i$ has value $b$ for all goods, then by the argument in the proof of Lemma~\ref{clm:ef-to-efx}, agent $i$ will be $\efx_0$ at the end of the algorithm. 
If there is at least one good which $i$ values as $a$, then $r_i$ is well-defined. 
Recall that by Lemma~\ref{clm:from-matching}, it holds that agent $i$ is allocated a good she values as $a$ in each of the rounds $1,2, \dots, r_i-1$. In round $r_i$ there are two cases:

\medskip

\noindent\emph{\underline{Case I:} Agent $i$ is allocated a good she values as $a$.} 
Then, at the end of round $r_i$, agent $i$ has total value $a r_i$ for her bundle and total value at most $a r_i$ for the bundle of any other agent. If agent $i$ did not become frozen at the end of round $r_i$ (which means she never will), then she will be $\efx_0$ towards all agents at the end of the algorithm by Lemma~\ref{clm:ef-to-efx}. 
Now consider the case where agent $i$  froze at the end of round $r_i$.
\begin{itemize}
\item Any agent $j$ who obtained a good that $i$ values as $a$ in round $r_i$ will also freeze at the end of round $r_i$ (by the definition of the set $F$ of frozen agents), and both $i$ and $j$ will re-enter in the same round later. 
Thus, since $i$ is envy-free towards $j$ up until round $r_i$ and both agents will never freeze again, $i$ will be $\efx_0$ towards $j$ at the end of the algorithm by Lemma~\ref{clm:ef-to-efx}.

\item For any agent $j$ who obtained a good that $i$ values as $b$, the value of $i$ for the bundle of $j$ at the end of round $r_i$ is at most $a(r_i-1) + b$. Hence, when agent $i$ re-enters (after the end of round $r_i+\lfloor a/b-1 \rfloor$), her value for $j$'s bundle can be at most $a(r_i-1) + b + b \lfloor a/b-1 \rfloor \leq a r_i$. Thus, agent $i$ will be envy-free towards agent $j$ at this point, and $\efx_0$ towards $j$ at the end of the algorithm by Lemma~\ref{clm:ef-to-efx}. 
If the algorithm terminates before agent $i$ re-enters (which could happen if $b=0$), then agent $i$ is envy-free towards $j$.
\end{itemize}

\medskip

\noindent\emph{\underline{Case II:} Agent $i$ is allocated a good she values as $b$.}
At the end of round $r_i$ her total value is $a(r_i-1)+b$, and will never freeze by Lemma~\ref{clm:from-matching}.

\begin{itemize}
\item For any agent $j$ who obtained a good that $i$ values as $b$ in round $r_i$, the value of $i$ for $j$'s bundle can be at most $a(r_i-1)+b$. Thus, agent $i$ is envy-free towards $j$ at the end of $r_i$, and $\efx_0$ at the end of the algorithm by Lemma~\ref{clm:ef-to-efx}.

\item Any agent $j$ who obtained a good that $i$ values as $a$ in round $r_i$ must have  frozen at the end of round $r_i$. If the algorithm terminates before agent $j$ re-enters (for example, if $b=0$), then agent $i$ is $\efx_0$ with respect to $j$, because the value of $i$ for $j$'s bundle is at most $a r_i$ and the least valuable good from $i$'s perspective is of value $a$. 
Otherwise, agent $j$ re-enters before the termination of the algorithm after the end of round $r_i+\lfloor a/b-1 \rfloor$, and the value of $i$ for $j$'s bundle is still at most $a r_i$, as $j$ did not receive any other good. Since $i$'s own value increased to $a(r_i-1)+b + b \lfloor a/b-1 \rfloor > a r_i - b$, her envy towards agent $j$ is at most $b$ at this point. 
Furthermore, $i$ has value $b$ for the remaining goods. 
Because in the last round we prioritize agents who have never gotten frozen (which means that if there are less goods than agents, $i$ has priority in getting an available good compared to $j$), by an argument similar to the one used in the proof of Lemma~\ref{clm:ef-to-efx}, it follows that the envy will still be at most $b$ at the end of the algorithm.
\end{itemize}
This completes the proof.
\end{proof}

\subsection{Interval-value instances} 
From our discussion thus far, it may seem like the difficultly of proving the existence of $\efx_0$ allocations is directly related to the number of different values that the agents have, but this is not entirely true. We will now show that the range between the lowest and the highest value also plays a very important role: for specific ranges, and independently of the number of values therein (which can be infinite), computing $\efx$ allocations can be achieved by very simple algorithms. 
In particular, we show that $\efx$ allocations exist for interval-instances in which the values of each agent $i$ are in some interval $[x_i,2x_i]$, $x_i \in \mathbb{R}_{>0}$, by using a simple modification of the \emph{round-robin algorithm}. 

According to this algorithm, we fix an ordering of the agents and then they simply pick their favorite unallocated good one by one, with respect to that ordering. This continues in rounds of $n$ goods each, until we reach a point where there are not enough goods for everyone. For this last round (if it exists), the agents pick in reverse order (see Algorithm~\ref{alg:MRR}).

\begin{algorithm}
\caption{Modified round-robin}
\begin{algorithmic}[1]
\State Input: Instance with $m=kn+\ell$, $k \geq 0$, $0\leq \ell < n$ 
\State $S=M$ 
\For{$r = 1, \dots, k$}
\For{$i = 1, \dots, n$}
\State $g \in \arg\max_{q \in S}v_i(q)$
\State $A_i = A_{i} \cup \{g\}$
\State $S = S\setminus \{g\}$
\EndFor
\EndFor
\For{$i = n, \dots, n-\ell+1$}
\State $g \in\arg\max_{q \in S}v_i(q)$
\State $A_i=A_i \cup \{g\}$
\State $S=S\setminus \{g\}$
\EndFor 
\State 
\Return {$\bfA=(A_1,\dots,A_n)$}
\end{algorithmic}
\label{alg:MRR}
\end{algorithm}

\begin{theorem}\label{thm:modifiedRR}
Given an interval-instance in which the values of agent $i$ are in the interval $[x_i,2x_i]$, $x_i \in \R_{> 0}$, \emph{Modified round-robin} computes an EFX allocation in polynomial time.
\end{theorem}

\begin{proof}
Let $\bfA=(A_1,\dots, A_n)$ be the allocation produced by Algorithm~\ref{alg:MRR}. First observe that if $m < n$ then the statement holds trivially. Hence, we assume that $m=kn+\ell$ for some $k \geq 1$ and $0 \leq \ell < n$. In this case, the first $n-\ell$ agents will get $k$ goods and the last $\ell$ agents will get $k+1$ goods. Consider now an agent $i$ and let $g_{ir}$ be the good that she gets in round $1 \leq r \leq k+1$. We will show that $i$ is $\efx$ towards any other agent $j$, by distinguishing between cases depending on whether $i$ selects before or after $j$ according to the main ordering of the algorithm.
\medskip 

\noindent\emph{\underline{Case I:} $i < j$.}
Agent $j$ either has the same number of goods as $i$, or one more good than $i$. 
\begin{itemize}
\item If both agents get $k$ goods, since $i$ always chooses her most-valuable good before $j$, we have that $v_i(g_{ir}) \geq v_i(g_{jr})$ for every $r \in [k]$, and thus $i$ does not envy $j$. 

\item If both agents get $k+1$ goods or agent $i$ has $k$ goods and agent $j$ has $k+1$ goods, let $g_{jt} \in A_j$ be the least-valuable good according to agent $i$, which $j$ gets during round $t \in [k+1]$. 
If $t=k+1$, then since $v_i(g_{ir}) \geq v_i(g_{jr})$ for every $r \in [k]$, we immediately obtain that $v_i(A_i) \geq v_i(A_j \setminus \{g_{i,k+1}\})$. 
If $t \leq k$, since $v_i(g_{it}) \geq v_i(g_{i,k+1})$ and $v_i(g_{ir}) \geq v_i(g_{jr})$ for every $r \in [k]\setminus\{t\}$, we again obtain that $v_i(A_i) \geq v_i(A_j \setminus \{g_{it}\})$.   
\end{itemize}
\medskip 

\noindent\emph{\underline{Case II:} $i > j$.}
Agent $j$ has either the same number of goods as $i$ or one less good than $i$. 
Once again, let $g_{jt} \in A_j$ be the least-valuable good according to agent $i$. Now observe that in general we have that $v_i(g_{ir}) \geq v_i(g_{j,r+1})$ for every $r \in [k-1]$. If $t = 1$, then the statement holds trivially. So, we assume that $t \geq 2$. 
\begin{itemize}
\item 
If agent $j$ gets $k$ goods, then we have that $v_i(g_{ir}) \geq v_i(g_{j,r+1})$ for every $1 \leq r \leq t-2$ and $v_i(g_{ir}) \geq v_i(g_{j,r+2})$ for every $t-1 \leq r \leq  k-2$. So far we have bounded all goods in $A_j \setminus \{g_it\}$ besides $g_{j1}$, using all goods in $A_i$ besides $g_{i,k-1}$ and $g_{i,k}$ (and also $g_{i,k+1}$ if it exists). Since all the values lie in the interval $[x_i,2x_i]$, we have that $v_i(\{g_{ik-1}) + v_i(g_{ik}\}) \geq v_i(\{g_{j1}\})$, which yields that $v_i(A_i) \geq v_i(A_j \setminus \{g_{it}\})$. 

\item
If both agents get $k+1$ goods, then we have that $v_i(g_{ir}) \geq v_i(g_{j,r+1})$ for every  $1 \leq r \leq t-2$ and $v_i(g_{ir}) \geq v_i(g_{j,r+2})$ for every $t-1\leq r \leq k-1$. Again, the values of the goods in $A_j \setminus \{g_it\}$ besides $g_{j1}$ have been bounded by the values of the goods in $A_i$ besides $g_{ik}$ and $g_{i,k+1}$. Since all the values lie in the interval $[x_i,2x_i]$, we have that $v_i(\{g_{ik}) + v_i(g_{i,k+1}\}) \geq v_i(\{g_{j1}\})$, and thus $v_i(A_i) \geq v_i(A_j \setminus \{g_{it}\})$.
\end{itemize}
This completes the proof.
\end{proof}


\section{MNW and the EFX-value}\label{sec:approx}
As we saw in Section~\ref{sec:mnw}, maximizing the Nash welfare does not yield an $\efx$ allocation in general. 
Here we take a different route and instead of considering exact $\efx$ allocations, we focus on approximation. 
We start by showcasing that maximizing the Nash welfare does not guarantee \emph{any} meaningful approximation of $\efx$ according to the current definition of approximation used in the literature~\citep{PlautR18,AmanatidisBM18,AMN2020,ChanCLW19maximin}.

\begin{defn}[$\alpha$-$\efx$ allocation]
For $\alpha \in (0,1]$, an allocation $\bfA$ is $\alpha$-$\efx$ if for every pair $i,j \in N$ and every good $g \in A_j$ such that $v_i(g)>0$, it holds that $v_i(A_i) \geq \alpha v_i(A_j \setminus \{g\})$.
\end{defn}

\noindent
Let $w > 1$ and $\varepsilon < \frac{1}{2w}$. Consider the following very simple instance with two agents, three goods, and values given in the table:
\begin{center}
\begin{tabular}{c?cccc}
 & $g_1$ & $g_2$ & $g_3$ \\
\noalign{\hrule height 1pt}
agent $1$ & $w$ & $0$ & $1/2$ \\
agent $2$ & $w$ & $1$ & $\varepsilon$ \\
\end{tabular}
\end{center} 
We first claim that the allocation $\bfA^*=(A_1=\{g_1,g_3\}, \allowbreak A_2= \allowbreak \{g_2\})$ is the only one that achieves the maximum Nash welfare of $w + 1/2$. Indeed, the Nash welfare of any allocation that gives $g_2$ to agent $1$ can only increase by moving $g_2$ to agent $2$, while any allocation other than $\bfA^*$ that gives $g_2$ to agent $2$ has Nash welfare either  $(w + 1)/2$ or $w+\varepsilon w < w + 1/2$. Notice, however, that  $\bfA^*$ is not $\efx$ since $v_2(A_2)=1 < w = v_2(A_1 \setminus \{g_3\})$. Instead, it is only $1/w$-$\efx$, an approximation factor that can be arbitrarily close to zero as $w$ becomes large. 

Nevertheless, $\bfA^*$ is not that far away from being an $\efx$ allocation! To see this, consider the allocation $\bfB = (B_1=\{g_1\},B_2=\{g_2,g_3\})$ that is obtained from $\bfA^*$ by only moving $g_3$ from agent $1$ to agent $2$. Clearly, agent $2$ is $\efx$ towards agent $1$, as the latter gets only one good. Moreover, the value  agent $2$  has now is $v_2(B_2)=1 + \varepsilon$, which is extremely close to the value $v_2(A_2)=1$ that agent $2$ has in $\bfA^*$. So, even though $\bfA^*$ is  $1/w$-$\efx$ because $v_2(A_2)$ is very low compared to $v_2(A_1 \setminus \{g_3\})$, $v_2(A_2)$ is actually very close to the value she would have in a \emph{nearby} $\efx$ allocation. Consequently, if we accept that agent $2$ considers the $\efx$ allocation $\bfB$ as fair, then she should also consider $\bfA^*$ as being {\em almost} fair.    

We say that the value $v_2(B_2) = 1 + \varepsilon$ is the {\em $\efx$-value} that agent $2$ can achieve by augmenting her bundle with a subset of goods from agent $1$ in order to create the closest-to-$\bfA^*$ (in terms of value) allocation $\bfB$ which she considers as $\efx$. Then, since $v_2(A_2) = \frac{1}{1+\varepsilon}v_2(B_2)$, agent $2$ achieves an approximation of $\frac{1}{1+\varepsilon}$ of her $\efx$-value. Let us now formalize these notions for any number of agents. 

\begin{defn}[$\efx$-value]
Let $\bfA = (A_1,\dots,A_n)$ be an allocation. For every pair of agents $i,j \in N$, let $X_{ij} \subseteq A_j$ be a set of goods such that $v_i(A_i\cup X_{ij}) \geq v_i(A_j \setminus (X_{ij} \cup \{g\}))$ for every $g \in A_j \setminus X_{ij}$ and $v_i(A_i \cup X_{ij})$ is minimized. Then, the {\em $\efx$-value} of agent $i$ is 
$$\xv_i(\bfA) = \max_{j \in N \setminus \{i\}} v_i(A_i \cup X_{ij}).$$
\end{defn}

\begin{defn}[$\alpha$-$\vefx$ allocation]
For $\alpha \in (0,1]$, an allocation $\bfA$ is $\alpha$-$\vefx$ if $v_i(A_i) \geq \alpha \xv_i(\bfA)$ for every $i \in N$.
\end{defn}

\noindent
We remark that using $\efx_0$ instead of $\efx$ in the above definitions does not make any difference since adding zeros does not affect the $\efx$-value. Furthermore, observe that a $1$-$\vefx$ allocation is an $\efx$ allocation but not necessarily an $\efx_0$ allocation.  

Our first technical result in this section illustrates the connection between approximate $\efx$ and $\vefx$ allocations. 

\begin{theorem}\label{thm:efx-efx}
For $\alpha\in (0,1)$, an $\alpha$-$\efx$ allocation is also an $\frac{\alpha}{1+\alpha}$-$\vefx$ allocation and this guarantee is tight.
On the other hand, an $\alpha$-$\vefx$ allocation is not guaranteed to be $\beta$-$\efx$, for any $\alpha, \beta \in (0,1)$. 
\end{theorem}

\begin{proof}
Let $\bfA$ be an $\alpha$-$\efx$ allocation such that $v_i(A_i) \geq \alpha \cdot v_i( A_j\setminus \{g\})$, where $v_i(g)\in \min_{q \in A_j}v_i(q)$, for every pair of agents $i,j \in N$. We can equivalently write this inequality as $(1+\alpha)v_i(A_i) \geq \allowbreak \alpha ( v_i(A_i) \allowbreak + v_i( A_j\setminus \{g\} )$, and since $\xv_i(\bfA) \leq v_i(A_i \cup A_j\setminus \{g\})$, we obtain that 
$$\frac{v_i(A_i)}{\xv_i(\bfA)} \geq \frac{v_i(A_i)}{v_i(A_i \cup A_j\setminus \{g\})} \geq \frac{\alpha}{1+\alpha}.$$ 
Furthermore, by moving all goods in $A_j\setminus \{g\}$ from $j$ to $i$, we obtain an allocation in which $j$ gets only one good, and consequently $i$ is $\efx$ towards $j$. Hence, $\bfA$ is indeed $\frac{\alpha}{\alpha+1}$-$ \vefx$.

For the upper bound, let $\alpha \in (0,1)$ and consider an instance in which some agent $i$ has values given by the following table:
\begin{center}
\begin{tabular}{c?cccc}
& $g_1$ & $g_2$ & $g_3$ \\
\noalign{\hrule height 1pt}
agent $i$ & $1$ & $1/\alpha$ & $1/\alpha$ \\
\end{tabular}
\end{center}
In the allocation $\bfA$ according to which $A_i=\{g_1\}$ and $A_j = \{g_2, g_3\}$ for some agent $j \neq i$, agent $i$ is not $\efx$ towards agent $j$. Since $v_i(A_i)=1$ and $v_i(A_j \setminus \{g_2\})=1/\alpha$, we have that $\bfA$ is $\alpha$-$\efx$. 
Now, by moving $g_2$ from $j$ to $i$, $i$ becomes $\efx$ towards $j$, and achieves a total value of $v_i(A_i) + v_i(g_2) = 1 + 1/\alpha$, yielding an approximation of $\frac{1}{1/\alpha + 1}=\frac{\alpha}{1+\alpha}$ on her $\efx$-value.

For the reverse relation, let $\alpha \in (0,1)$ and $\gamma > \frac{1-\alpha}{\alpha}$. Consider an instance in which the values of some agent $i$ are:
\begin{center}
\begin{tabular}{c?cccc}
 & $g_1$ & $g_2$ & $g_3$ \\
\noalign{\hrule height 1pt}
agent $i$ & $1$ & $\gamma$ & $(1-\alpha)/\alpha$ \\
\end{tabular}
\end{center}
Let $\bfA$ be the allocation according to which $A_i=\{g_1\}$ and $A_j = \{g_2, g_3\}$ for some agent $j \neq i$.
Since $v_i(A_i) = 1$, agent $i$ can become $\efx$ towards $j$, by acquiring $g_3$, in which case she achieves a value of $1/\alpha$. Hence, $\bfA$ is $\alpha$-$\vefx$. However, since $v_i(A_j \setminus \{g_3\}) = \gamma$, $\bfA$ is only $1/\gamma$-$\efx$; the approximation can be arbitrarily close to $0$ as $\gamma$ becomes  large.
\end{proof}

Even though maximizing the Nash welfare may not yield a $\beta$-EFX for any $\beta \in (0,1)$ as we showed above, it is guaranteed to produce a constant $\vefx$ allocation.

\begin{theorem}\label{MNWapprox}
Any maximum Nash welfare allocation $\bfA^*$ is $1/2$-$\vefx$, and this bound is tight.
\end{theorem}

By using arguments similar to those in the proof of Theorem \ref{thm:efx-efx}, we can show that any EF1 allocation is $1/2$-$\vefx$. Then, Theorem \ref{MNWapprox} follows from the result of \citet{CaragiannisKMPS19} about MNW implying EF1. Below, we present a direct and self-contained proof of Theorem \ref{MNWapprox}.

\begin{proof}
Consider a pair of agents $i$ and $j$. 
Let $A_i=\{a_1, ..., a_{|A_i|}\}$ and $A_j=\{b_1, ..., b_{|A_j|}\}$ be the sets of goods allocated to $i$ and $j$ according to $\bfA^*$, respectively. Without loss of generality, we may assume that $0 \leq v_i(b_1) \leq ... \leq v_i(b_{|A_j|})$. Let $S = \{b_1, ..., b_{\lambda}\}$, $\lambda < |A_j|$ be the subset of least-valued goods of $A_j$ such that $i$ is $\efx$ towards $j$ when given the set $S$ (on top of $A_i$), but is not $\efx$ towards $j$ when given only the set $S\setminus \{b_\lambda\}$. 
We are going to show that $v_i(A_i) \geq v_i(S) \Leftrightarrow 2v_i(A_i) \geq v_i(A_i \cup S)$. 
If this is true, then since $\xv_i(\bfA) \leq v_i(A_i \cup S)$, $\bfA^*$ must be $1/2$-$\vefx$. 

Assume towards a contradiction that $v_i(A_i) < v_i(S)$. 
By its definition, $S$ is such that 
$$v_i(A_i) + v_i(S \setminus \{b_\lambda\}) < v_i(A_j \setminus (S \setminus \{b_\lambda\})) - v_i(b_\lambda)$$ 
or, equivalently, 
$$v_i(A_i) + v_i(S \setminus \{b_\lambda\}) < v_i(A_j \setminus S).$$ 
Also $v_i(S \setminus \{b_\lambda\}) \geq 0$ implies that 
$v_i(A_i) < v_i(A_j \setminus S)$.
Since $v_j(S) + v_j(A_j \setminus S) = v_j(A_j)$, it must be the case that one of $v_j(S)$ and $v_j(A_j \setminus S)$ is at least $\frac{v_j(A_j)}{2}$, while the other is at most this much. 
\begin{itemize}
\item 
If $v_j(S) \leq \frac{v_j(A_j)}{2}$ and $v_j(A_j \setminus S) \geq \frac{v_j(A_j)}{2}$, then we define a new allocation in which all goods in $S$ are moved from agent $j$ to agent $i$. By our assumption that $v_i(A_i) < v_i(S)$, the product of the new values of the two agents is
\begin{align*}
\big( v_i(A_i) + v_i(S) \big) \cdot v_j(A_j \setminus S) > 2 v_i(A_i) \cdot \frac{v_j(A_j)}{2} = v_i(A_i) \cdot v_j(A_j). 
\end{align*}
Since the sets of goods given to the other agents have not been altered, the new allocation has strictly more Nash welfare than $\bfA^*$, a contradiction.

\item 
If $v_j(S) \geq \frac{v_j(A_j)}{2}$ and $v_j(A_j \setminus S) \leq \frac{v_j(A_j)}{2}$, then we define another allocation in which all goods in $A_j \setminus S$ are moved from $j$ to $i$. Since $v_i(A_i) < v_i(A_j \setminus S)$, the product of the new values of the two agents becomes
\begin{align*}
\big( v_i(A_i) + v_i(A_j \setminus S) \big) \cdot v_j(S) > 2 v_i(A_i) \cdot \frac{v_j(A_j)}{2} = v_i(A_i) \cdot v_j(A_j), 
\end{align*}
again yielding a contradiction.
\end{itemize}
Consequently, it must be $v_i(A_i) \geq v_i(S)$, meaning that $\bfA^*$ is $1/2$-$\vefx$. 

For the upper bound, consider again the instance used in Section~\ref{sec:mnw} to show that an MNW allocation may not be $\efx$ for $3$-value instances. For ease of reference, we repeat the table containing the values of the two agents for the three goods here:
\begin{center}
\begin{tabular}{c?cccc}
 & $g_1$ & $g_2$ & $g_3$ \\
\noalign{\hrule height 1pt}
agent 1 & $1-\varepsilon$ & $1$ & $1+\varepsilon$ \\
agent 2 & $1$ & $1-\varepsilon$ & $1+\varepsilon$ \\
\end{tabular}
\end{center}
As we argued in Section~\ref{sec:mnw}, the allocations
$\bfA_1 = (\{g_2,g_3\}, \{g_1\})$ and $\bfA_2 = (\{g_2\}, \{g_1,g_3\})$ are the only Nash welfare maximizing ones. 
Both of these are $\frac{1}{2-\varepsilon}$-$\vefx$ since the envious agent can get the least-valued good from the other agent (worth $1-\varepsilon$) and become $\efx$. 
\end{proof}

\section{Directions for Future Work}
\label{sec:open}
We studied the connection between two celebrated notions, that of maximum Nash welfare and envy-freeness up to any good. We showed that a maximum Nash welfare allocation is always $\efx_0$ for $2$-value instances, while this implication is no longer true for $k$-value instances with $k \geq 3$. The first question that our work leaves open is whether it is possible to compute in polynomial-time an allocation that maximizes the Nash welfare for $2$-value instances.

Nevertheless, for $2$-value instances we presented a polynomial-time algorithm for computing an $\efx_0$ allocation. Due to its novelty, we believe that the idea of repeatedly computing maximum matchings and freezing certain agents whenever they acquire too much value (compared to other agents), might be a stepping stone for proving the existence of $\efx_0$ more generally. That being said, while generalizing our algorithm to $k$-value instances with $k \geq 3$ definitely deserves further investigation, it does seem to be a highly non-trivial task.

Finally, going beyond exact MNW or EFX allocations, we discussed the connection between MNW and approximate EFX allocations. While an MNW allocation does not necessarily provide any meaningful guarantee according to the commonly used definition of approximation, we showed that it does under a new interpretation of approximation via the $\efx$-value. A natural question is whether one can design polynomial-time algorithms with strong approximation guarantees for both the $\efx$-value and the Nash welfare.

\bibliographystyle{named}
\bibliography{mnw-efx-references-full}

\end{document}